\documentclass[11pt,a4paper,oneside,onecolumn]{article}

\usepackage{cite}
\usepackage[english]{babel}
\usepackage[ansinew]{inputenc}
\usepackage{epsfig,float,latexsym,enumerate,amssymb,amsmath,amsfonts,color}
\usepackage{subfigure}

\newcommand{\x}[1]{x_{#1}}
\newcommand{\y}[1]{y_{#1}}
\def\eps{\varepsilon}

\newtheorem{theorem}{Theorem}[section]

\newtheorem{claim}[theorem]{Claim}

\newtheorem{corollary}[theorem]{Corollary}

\newtheorem{lemma}[theorem]{Lemma}

\newtheorem{proposition}[theorem]{Proposition}

\def\QED{\ensuremath{{\square}}}
\def\markatright#1{\leavevmode\unskip\nobreak\quad\hspace*{\fill}{#1}}

\newenvironment{proof}
{\begin{trivlist}\item[\hskip\labelsep{\bf Proof.}]}
{\markatright{\QED}\end{trivlist}}

\begin{document}

\title{Locating a single facility and a high-speed line}

\author{
	J. M. D\'{\i}az-B\'{a}\~{n}ez
	\thanks{Departamento Matem\'{a}tica Aplicada II, Universidad de Sevilla, Espa\~na. 
	{\tt \{dbanez,iventura\}@us.es}} \and 
	M. Korman
	\thanks{Universitat Polit\`{e}cnica de Catalunya (UPC), Barcelona, Espa\~na.
	{\tt mkormanc@ulb.ac.be}} \and
	P. P\'{e}rez-Lantero
	\thanks{Escuela de Ingenier\'ia Civil en Inform\'atica, Universidad de Valpara\'{i}so, Chile.
	{\tt pablo.perez@uv.cl}} \and
	I. Ventura\footnotemark[1]
} 

\maketitle

\begin{abstract}
In this paper we study a facility location problem in the plane in
which a single point (facility) and a rapid transit line (highway)
are simultaneously located in order to minimize the total travel
time from the clients to the facility, using the $L_1$ or Manhattan
metric. The rapid transit line is given by a segment with any length
and orientation, and is an alternative transportation line
that can be used by the clients to reduce their travel time to the
facility. We study the variant of the problem in which clients can
enter and exit the highway at any point. We provide an $O(n^3)$-time
algorithm that solves this variant, where $n$ is the number
of clients. We also present a detailed characterization of the solutions,
which depends on the speed given in the highway.
\end{abstract}
    

%
%

\section{Introduction}\label{section:intro}

Suppose we are given a set of clients represented as a set of points in
the plane, and a service facility represented as a point to which
all clients have to move. Every client can reach the facility
directly, or use an alternative rapid transit line called highway in order to reduce the
travel time. The highway is a straight line segment of arbitrary
orientation.
If a client moves
directly to the facility, it moves at unit speed and the distance
traveled is the Manhattan or $L_1$ distance to the facility. In the
case where a client uses the highway, it travels the $L_1$ distance
at unit speed to one point of the highway, traverses with a speed
$v>1$ the Euclidean distance to other highway point, and
finally travels the $L_1$ distance from that point to the facility
at unit speed. All clients traverse the highway at the same speed.
The highway is used by a client point whenever it saves time to reach
the facility.
Given the set of points representing the clients, the facility
location problem consists in determining at the same time the
facility point and the highway in order to minimize the total
weighted travel time from the clients to the facility. The weighted
travel time of a client is its travel time multiplied by a weight
representing the intensity of its demand.

Recent papers have dealt with geometric problems considering
travelling distances as a combination of planar and network distances.
Carrizosa and Rodr\'iguez-Ch\'ia~\cite{carrisoza1997} introduced the
$p$-facility min-sum location problem on the plane with a metric 
induced by a gauge and a finite set of rapid transit lines giving the network distance.
This problem was further developed by 
Gugat and Pfeiffer~\cite{gugat2007}, and Pfeiffer and Klamroth~\cite{pfeiffer2008}.
Brimberg et al.~\cite{brimberg2003,brimberg2005} studied 
the location of a new single facility
considering given regions of distinct distance measures.
All these papers consider the well-known
Weber problem under a new metric.
Other papers can be found in the context of the combination of the 
$L_1$ distance with the network distance.
Abellanas et al.~\cite{abellanas03}
introduced the time metric model: Given an
underlying metric, the user can travel at speed $\nu(h)$ when moving
along a highway $h$ or unit speed elsewhere. The particular case in
which the underlying metric is the $L_1$ metric and all highways are
axis-parallel segments of the same speed, is called the {\em city
metric}~\cite{aichholzer02}. 
In the scenario of setting up an optimal distance network
that minimizes the maximum travel time among a
set of points, several problems have been recently investigated in detail~\cite{ahn07,aloupis10,cardinal08}. Other similar and more general models
were studied by Korman and Tokuyama~\cite{korman08}.
See~\cite{phdkorman} and~\cite{dbanezMS04} for surveys on highway location 
and extensive facility location problems, respectively.

A similar problem consisting in simultaneously locating a
service facility point and a high-speed line (i.e.\ highway) 
of fixed length was recently studied by
Espejo and Rodr\'{i}guez-Ch\'{i}a~\cite{espejo11} and D\'iaz-B\'a\~nez et
al.~\cite{diaz-banez11}. The authors considered that highway is a 
\emph{turnpike}~\cite{korman08}, that is, clients can enter
and exit the highway only at the endpoints. 
A first solution was introduced by Espejo and Rodr\'{i}guez-Ch\'ia, and
after that an improved solution was given by D\'iaz-B\'a\~nez et al. 
The problem aims to minimize
the total weighted travel time from the demand points to the facility
service, and can be solved in $O(n^3)$ time~\cite{diaz-banez11}, where 
$n$ is the number of clients.
D\'iaz-B\'a\~nez et al.~\cite{diaz-banez11-3}
continued the study of this variant by considering the min-max
optimization criterion. They minimize the maximum time distance from
the clients to the facility point.

In this paper we study a related problem in which the length of the
highway is variable, that is, it is not fixed in advance as part
of the input of the problem, and clients can enter and
exit the highway at any point. Due to the latter condition, highway is called
\emph{freeway}~\cite{korman08}. Since both entering and leaving the highway are
allowed at all its points, then the structure (i.e.\ highway) is continuously
integrated in the plane.
We minimize the total weighted
transportation time from the demand points to the facility. The problem
of locating a min-max freeway of fixed length was solved by
D\'iaz-B\'a\~nez et al.~\cite{diaz-banez11-3} in $O(n\log n)$ time.

The following notation is introduced in order to formulate the problem.
Let $S$ be the set of $n$ demand points, $f$ be the service facility
point, $h$ be the highway, and $v>1$ be the speed in which demand
points move along $h$. Given a demand point $p$, $w_p>0$ denotes the
weight of $p$. The travel time between a demand point $p$ and the
service facility $f$, denoted by $d_h(p,f)$, is equal to:
\begin{equation}\label{eq1}
\min\left\{
\begin{array}{l}
\|p-f\|_1,\\
\min_{q_1,q_2\in
h}\left\{\|p-q_1\|_1+\frac{\|q_1-q_2\|_2}{v}+\|q_2-f\|_1\right\}
\end{array}
\right.
\end{equation}

The problem can be formulated as follows:

\medskip

{\bf The Freeway and Facility Location problem (FFL-problem)} Given
a set $S$ of $n$ demand points, the weight $w_p>0$ of each point $p$ of
$S$, and fixed speed $v>1$, locate the facility point $f$ and the
highway $h$ in order to minimize the next function:
\begin{equation}\label{eqnobj}
\Phi(f,h):=\sum_{p\in S}w_p \cdot d_h(p,f).
\end{equation}

\paragraph{Our results.}
We first show that there exist
optimal solutions of the FFL-problem in which the highway $h$ has
infinite length and the facility point $f$ is located on $h$. We
then consider only optimal solutions satisfying these properties. We
second show that for all demand points $p$ the shortest path
from $p$ to the facility point $f$ has one of three possible shapes:
(a) $p$ moves directly to $f$, (b) $p$ first moves vertically to
reach $h$ and after that moves along $h$ to reach $f$, and (c) $p$
first moves horizontally to reach $h$ and after that moves along $h$
to reach $f$. For each demand point, the shape of its shortest time
path to $f$ depends on both the speed $v$ in which demand points
moves along $h$ and the slope of $h$. This discretization on the
shortest path shapes allows us to simplify the expression of
$d_h(p,f)$ and then to obtain a clear
expression of the objective function $\Phi(f,h)$. 
Using geometric obervations, we reduce the search space of the optimal solutions. This is
done by considering
the grid $G$ defined by all axis-parallel lines passing through the
demand points. 
We prove the existence of optimal solutions $(f,h)$
that satisfy one of the following two properties: (1) $h$ passes
through a demand point and $f$ belongs to a line of grid $G$, and
(2) $f$ is a vertex of $G$. The discretization of the search space
permits us to obtain the main result of this paper, a general
$O(n^3)$-time algorithm to solve the FFL-problem. Our algorithm
divides the search into two cases that correspond to the above two
properties. As a surprising result, we prove that when speed $v$ is
greater than $\frac{3\sqrt{2}}{4}\approx 1.060660172$ the algorithm
can avoid the search of optimal solutions satisfying property (2)
because in that case there always exists an optimal solution which
holds property (1). This result simplifies the algorithm when speed
exceeds that bound. We finally present three examples,
two of them showing that when speed is increased and we keep the
same configuration of demand points the shapes of the shortest time
paths can change. A third example shows that when speed is less
than $\frac{3\sqrt{2}}{4}$, there exist configurations in which the
optimal solution satisfies property (2).

\paragraph{Outline.}
The discretization on the shapes of the shortest paths from the
demand points to the facility is stated in Section~\ref{section:preliminaries}.
In Section~\ref{section:discretization} we show how the search
space of optimal solutions can be reduced.
In Section~\ref{section:algorithm-free} the algorithm to solve 
the FFL-problem is presented and in Section~\ref{section:refinement} we give the
refinement of it. In Section~\ref{section:experimental}, the examples are presented.
Finally, in Section~\ref{section:conclu}, we present the conclusions
and further research.

\section{Discretization of the shortest paths}
\label{section:preliminaries}

Any solution to our problem will be encoded by a pair of elements
$(f,h)$, where $f$ is the facility point and $h$ is the highway.
Given $f$ and $h$, we say that a demand point $p$ does not use $h$
(or goes directly to $f$) if $d_{h}(p,f)$ is equal to $\|p-f\|_1$.
Otherwise we say that $p$ uses $h$. Given any point $u$ of the plane,
let $\x{u}$ and $\y{u}$ denote the $x-$ and $y-$coordinates of $u$,
respectively.
\begin{claim}\label{claim1}
Let $p_1$ and $p_2$ be demand points using the highway $h$ such that 
they move in contrary directions along $h$. Let segments $s_1,s_2\subseteq h$
denote the portions of $h$ traversed by $p_1$ and $p_2$, respectively.
Segments $s_1$ and $s_2$ have disjoint interiors.
\end{claim}
\newpage
\begin{proposition}\label{prop:fac-in-highway}
There exists an optimal solution of the
FFL-problem in which the facility point is located on the highway.
\end{proposition}
\begin{proof}
Let $(f,h)$ denote an optimal solution of the FFL-problem and
suppose that $f$ does not belong to $h$. Let $h'$ be a translation of $h$
such that $f$ belongs to $h'$. We select $h'$ so that to satisfy a condition that
will be stated later. Let $p$ be a demand
point. If $p$ does not use $h$ then:
\begin{equation}
d_{h'}(p,f)\leq\|p-f\|_1=d_h(p,f).
\end{equation}
Otherwise, if $p$ uses $h$, let $q_1,q_2\in h$ be points such that
$$d_h(p,f)=\|p-q_1\|_1+\frac{\|q_1-q_2\|_2}{v}+\|q_2-f\|_1.$$
Let $q_3$ be the point $f+(q_1-q_2)$, which belongs to the line
containing $h'$. Observe from
Claim~\ref{claim1}
that we can select $h'$ so that point $q_3$ belongs to $h'$ for every
demand point $p$ using $h$. Then, by using the triangular inequality with the
$L_1$ metric, we obtain:
\begin{eqnarray}
\nonumber
 d_{h'}(p,f) & \leq & \|p-q_3\|_1+\frac{\|q_3-f\|_2}{v} \\
\nonumber
             &   =  & \|p-q_1+(q_2-f)\|_1+\frac{\|q_1-q_2\|_2}{v} \\
\nonumber
             & \leq & \|p-q_1\|_1+\frac{\|q_1-q_2\|_2}{v}+\|q_2-f\|_1\\
             &   =  & d_h(p,f)\label{eq2}
\end{eqnarray}
From equations (\ref{eq1}) and (\ref{eq2}) we have
$\Phi(f,h')=\sum_{p\in S}w_p \cdot d_{h'}(p,f)\leq\sum_{p\in
S}w_p\cdot d_h(p,f)=\Phi(f,h)$. Then the pair $(f,h')$ must be an
optimal solution and the result thus follows.
\end{proof}

Results similar to Propostion~\ref{prop:fac-in-highway}, stating that
the facility point belongs to the corresponding highway, can be 
found in~\cite{diaz-banez11-3,espejo11}. 
Observe from equations~(\ref{eqnobj}) and~(\ref{eq1})
that there always exists an
optimal solution $(f,h)$ to the FFL-problem in which the length of
$h$ is infinite.
We then assume from this
point forward that every solution satisfies that the highway is a straight
line and the facility point belongs to the highway.
Observe that this assumption does not have negative consequences,
due to the fact that in practice a highway of infinite length is not
possible. In fact, if a solution $(f,h)$ to the problem is such that
$h$ is a straight line, then $(f,h'')$ is also an optimal solution,
where $h''\subset h$ is the segment of minimum length 
such that every demand point both enter and exit $h$ on a point of $h''$.

Let $\alpha$ always denote the non-negative angle of the highway
with respect to the positive direction of the $x$-axis. Unless
otherwise specified, we assume $0\leq\alpha\leq\frac{\pi}{4}$.
Observe that if $\alpha>\frac{\pi}{4}$ we can, by properties of
$L_1$ and $L_2$ metrics, modify the coordinate system so that angle
$\alpha$ satisfies $0\leq\alpha\leq\frac{\pi}{4}$.

Given highway $h$ and a demand point $p$, let $p'$ be the
intersection point between $h$ and the vertical line passing through
$p$. Similarly, let $p''$ be the intersection point between $h$ and
the horizontal line passing through $p$. Let $h_{p'}$ denote the
half-line contained in $h$ that emanates from $p'$ and does not
contain $p''$, and $h_{p''}$ denote the half-line contained in $h$
that emanates from $p''$ and does not contain $p'$. Notice from the
assumption $0\leq\alpha\leq\frac{\pi}{4}$ that given $h$ and a
demand point $p$, $p'$ is the nearest point to $p$ on $h$ under the
$L_1$ metric.

The next lemma characterizes the way in which demand points move
optimally to the facility.
Let $\varphi_v=\frac{\pi}{4}-\arcsin\left(\frac{\sqrt{2}}{2v}\right)$. Since $v>1$
we have $0<\varphi_v<\frac{\pi}{4}$.

\begin{lemma}\label{lemma:way-of-move}
Given the highway $h$ and the facility point $f$ located on $h$, the
following holds for all demand points $p$.
If $0\leq\alpha\leq\varphi_v$ then $p$ moves first to $p'$ and after
that moves to $f$ using $h$.
Otherwise, if $\varphi_v<\alpha\leq\frac{\pi}{4}$, the next
statements are true:
\begin{itemize}
  \item[$($a$)$] If $f\in h_{p'}$ then $p$ moves first to $p'$ and after that
moves to $f$ using $h$.
  \item[$($b$)$] If $f\in h_{p''}$ then $p$ moves first to $p''$ and after that
moves to $f$ using $h$.
  \item[$($c$)$] If $f\in h\setminus (h_{p'}\cup h_{p''})$ then $p$ moves
directly to $f$ without using $h$.
\end{itemize}
\end{lemma}

\begin{proof}
Let $p$ be a demand point and assume w.l.o.g.\ that $p$ is below
$h$. Consider the function $g(u):=\|p-u\|_1+\frac{\|u-f\|_2}{v}$ for
all $u\in h$. Notice that $g$ is convex because it is a sum of two
convex functions. Then any local minimum of $g$ is a global minimum.
Let $\theta=\frac{\pi}{4}-\alpha$.

Suppose $0\leq\alpha\leq\varphi_v$. Given $\eps>0$ small enough, let
$u_1\in h\setminus h_{p'}$ and $u_2\in h_{p'}$ be the points such that
$\|p-u_1\|_1=\|p-u_2\|_1=\|p-p'\|_1+\sqrt{2}\eps$. Refer to
Fig.~\ref{fig:movement-2}.
\begin{figure}[h]
  \centering
  \includegraphics[width=7.0cm]{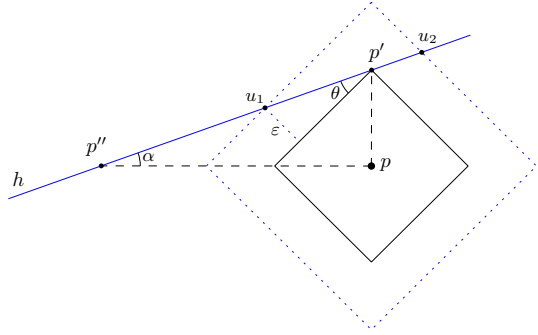}
  \caption{\small{Proof of Lemma~\ref{lemma:way-of-move}. The boundary of the square
    represented with solid lines is the set of points $u$ such that
    $\|p-u\|_1=\|p-p'\|_1$, and the perimeter of the square represented with dotted
    lines is the set of points $u$ such that $\|p-u\|_1=\|p-p'\|_1+\sqrt{2}\eps$.}}
  \label{fig:movement-2}
\end{figure}

Then we have the following:
\begin{eqnarray}
\nonumber
g(u_1)-g(p') & = & \sqrt{2}\eps +\frac{\|u_1-f\|_2-\|p'-f\|_2}{v}\\
\nonumber
  & \geq & \sqrt{2}\eps -\frac{\|p'-u_1\|_2}{v} \\
\nonumber
  &=& \eps\left(\sqrt{2} -\frac{1}{\sin\theta \cdot v}\right)\\
\nonumber
 & \geq & \eps\left(\sqrt{2} -\frac{1}{\sin(\frac{\pi}{4}-\varphi_v) \cdot v}\right)\\
\nonumber
 & = & \eps\left(\sqrt{2}
\nonumber
-\frac{1}{\sin\left(\arcsin\left(\frac{\sqrt{2}}{2v}\right)\right) \cdot  v}\right)\\
  & = & 0\label{eq3}
\end{eqnarray}
\begin{eqnarray}
\nonumber
g(u_2)-g(p') & = & \sqrt{2}\eps +\frac{\|u_2-f\|_2-\|p'-f\|_2}{v}\\
\nonumber
  & \geq & \sqrt{2}\eps -\frac{\|p'-u_2\|_2}{v} \\
\nonumber
  &   =  & \eps\left(\sqrt{2} -\frac{1}{\cos\theta \cdot v}\right)\\
\nonumber
 & \geq & \sqrt{2}\eps\left(\frac{v-1}{v}\right) \\
  & > & 0\label{eq4}
\end{eqnarray}
From equations (\ref{eq3}) and (\ref{eq4}) we conclude that $g(p')$
is the minimum of $g$. Therefore we have $d_h(p,f)=g(p')$ and the
first part of the lemma thus follows.

Suppose now $\varphi_v<\alpha\leq\frac{\pi}{4}$. Then we have three
cases:

{\em Case 1}: $f\in h_{p'}$. On one hand we have $g(u)>g(p')$ for
all points $u\in h\setminus h_{p'}$. On the other hand, if $\eps>0$
is small enough and $u_2\in h_{p'}$ is the point
such that $\|p-u_2\|_1=\|p-p'\|_1+\sqrt{2}\eps$, then
$g(u_2)-g(p')=\eps\left(\sqrt{2} -\frac{1}{\cos\theta \cdot
v}\right)>0$. Therefore, $d_h(p,f)=g(p')$ and statement (a) follows.

{\em Case 2}: $f\in h_{p''}$. Let $\eps>0$ be a small enough value.
On one hand, if $u_1\in h_{p''}$ is the point
such that $\|p-u_1\|_1=\|p-p''\|_1+\sqrt{2}\eps$, then
$g(u_1)-g(p'')=\eps\left(\sqrt{2} -\frac{1}{\cos\theta\cdot
v}\right)>0$. 
On the other hand, if $u_2\in h\setminus h_{p''}$ is the point
such that $\|p-u_2\|_1=\|p-p''\|_1-\sqrt{2}\eps$, then
\begin{eqnarray*}
g(u_2)-g(p'') & = & \eps\left(\frac{1}{\sin\theta\cdot v}-\sqrt{2}\right)\\
 & > & \eps\left(\frac{1}{\sin\left(\frac{\pi}{4}-\varphi_v\right)\cdot
v}-\sqrt{2}\right)\\
 & = & 0
\end{eqnarray*}
Therefore, $d_h(p,f)=g(p'')$ and statement (b)
follows.

{\em Case 3}: $f\in h\setminus (h_{p'}\cup h_{p''})$. If $\theta=0$
then $f$ is one of the nearest points to $p$ on $h$ by considering
the $L_1$ metric. Thus
$g(u)=\|p-u\|_1+\frac{\|u-f\|_2}{v}\geq\|p-u\|_1\geq\|p-f\|_1=g(f)$.
Otherwise, if $\theta>0$, we proceed as follows. On one hand we have
$g(u)>g(f)$ for all points $u\in h$ to the left of $f$. On the other
hand, if $\eps>0$ is small enough and $u_2\in h$ is the nearest
point to $f$ satisfying $\|p-u_2\|_1=\|p-f\|-\sqrt{2}\eps$, then
$g(u_2)-g(f)=\eps\left(\frac{1}{\sin\theta\cdot v}-\sqrt{2}\right)>0$.
Therefore, $d_h(p,f)=g(f)$ and statement (c) follows.
\end{proof}
Fig.~\ref{fig:movement} illustrates Lemma~\ref{lemma:way-of-move}.
\begin{figure}[h]
  \centering
  \includegraphics[width=15cm]{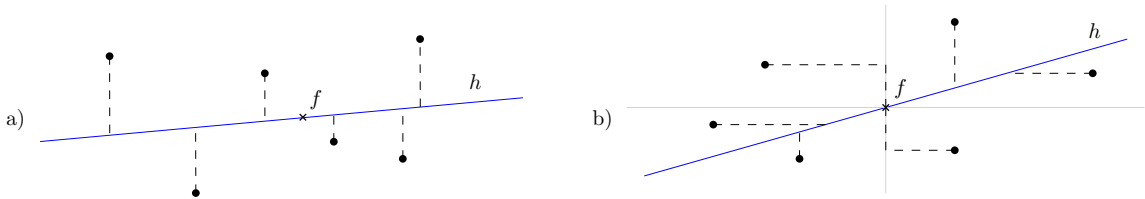}
  \caption{\small{a) If $0\leq\alpha\leq\varphi_v$ then all points
  move vertically to $h$.
  b) if $\varphi_v<\alpha\leq\frac{\pi}{4}$ then some points move
  vertically to $h$, some other points move horizontally, and the rest of the points move
  directly to $f$.}}
  \label{fig:movement}
\end{figure}
Because of Lemma~\ref{lemma:way-of-move}, the travel time $d_h(p,f)$ between a
demand point $p$ and facility $f$ simplifies to:
\begin{equation}\label{eq1}
\min\left\{
\begin{array}{l}
\|p-f\|_1,\\
\|p-p'\|_1+\frac{\|p'-f\|_2}{v},\\
\|p-p''\|_1+\frac{\|p''-f\|_2}{v}
\end{array}
\right.
\end{equation}

Given a solution $(f,h)$ to the
FFL-problem, we can always partition the set $S$ of demand points
into there sets $S_1:=S_1(f,h)$, $S_2:=S_2(f,h)$, and
$S_3:=S_3(f,h)$ as follows. Set $S_1$ contains the points $p\in S$
such that $x_p\leq x_f$ and $y_p\geq y_f$ and the points $q\in S$
such that $x_q\geq x_f$ and $y_q\leq y_f$. Set $S_2$ contains the
points $p\in S$ such that $x_p<x_f$ and $p$ is either on or below
$h$, and the points $q\in S$ such that $x_q> x_f$ and $q$ is either
on or above $h$. Set $S_3$ is equal to $S\setminus (S_1\cup S_2)$.
It is straightforward to obtain what follows (refer to Fig.~\ref{fig:expression-detail}).
\begin{figure}[h]
  \centering
  \includegraphics[width=15cm]{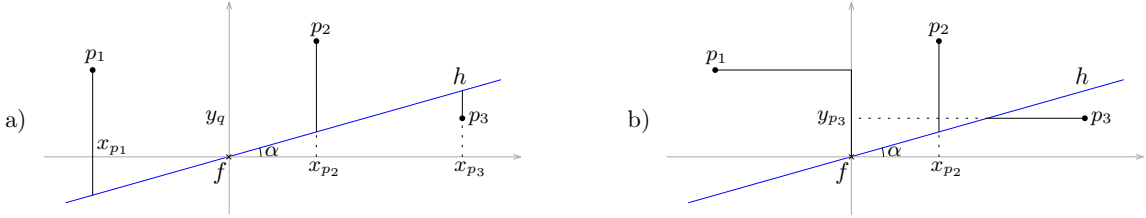}
  \caption{\small{Points $p_1$, $p_2$, and $p_3$ belong to $S_1$, $S_2$, and $S_3$,
  respectively. In case a) we have $0\leq\alpha\leq\varphi_v$.
  In case b) we have $\varphi_v<\alpha\leq\frac{\pi}{4}$.}}
  \label{fig:expression-detail}
\end{figure}

If $0\leq\alpha\leq\varphi_v$ then $d_h(p,f)$ is equal to:
\begin{equation*}
\left\{
\begin{array}{rr}
  |\y{p}-\y{f}|+|\x{p}-\x{f}|\tan\alpha+
  \frac{|\x{p}-\x{f}|}{\cos\alpha\cdot v} & \text{ if }p\in S_1\\
  |\y{p}-\y{f}|-|\x{p}-\x{f}|\tan\alpha+
  \frac{|\x{p}-\x{f}|}{\cos\alpha\cdot v} & \text{ if }p\in S_2\\
  -|\y{p}-\y{f}|+|\x{p}-\x{f}|\tan\alpha+
  \frac{|\x{p}-\x{f}|}{\cos\alpha\cdot v} & \text{ if }p\in S_3
\end{array}
\right.
\end{equation*}
and the objective function $\Phi(f,h)$ equals:
\begin{eqnarray}
\nonumber 
&&\sum_{p\in S_1}w_p|\y{p}-\y{f}|+\sum_{p\in S_2}w_p|\y{p}-\y{f}|-\sum_{p\in S_3}w_p|\y{p}-\y{f}|+\\
\nonumber
& &
\tan\alpha\left(\sum_{p\in S_1}w_p|\x{p}-\x{f}|-\sum_{p\in
S_2}w_p|\x{p}-\x{f}|+\sum_{p\in S_3}w_p|\x{p}-\x{f}|\right)+\\
&&\frac{1}{\cos\alpha\cdot v}\sum_{p\in S}w_p|\x{p}-\x{f}|\label{eq7}
\end{eqnarray}
Otherwise, if $\varphi_v<\alpha\leq\frac{\pi}{4}$, then
$d_h(p,f)$ is equal to:
\begin{equation*}
\left\{
\begin{array}{cr}
  |\x{p}-\x{f}|+|\y{p}-\y{f}| & \text{ if }p\in S_1\\
  |\y{p}-\y{f}|-|\x{p}-\x{f}|\tan\alpha+\frac{|\x{p}-\x{f}|}{\cos\alpha\cdot v} & \text{
if }p\in S_2\\
  |\x{p}-\x{f}|-|\y{p}-\y{f}|\cot\alpha+\frac{|\y{p}-\y{f}|}{\sin\alpha\cdot v} & \text{
if }p\in S_3
\end{array}
\right.
\end{equation*}
and $\Phi(f,h)$ equals:
\begin{eqnarray}
\nonumber
&&\sum_{p\in
S_1}w_p\left(|\x{p}-\x{f}|+|\y{p}-\y{f}|\right)+\sum_{p\in
S_2}w_p|\y{p}-\y{f}|+\sum_{p\in S_3}w_p|\x{p}-\x{f}|+\\
&&\left(\frac{1}{\cos\alpha\cdot v}-\tan\alpha\right)\sum_{p\in S_2}w_p|\x{p}-\x{f}|+
\left(\frac{1}{\sin\alpha\cdot
v}-\cot\alpha\right)\sum_{p\in S_3}w_p|\y{p}-\y{f}|\label{eq6}
\end{eqnarray}

\section{Reducing the search space}\label{section:discretization}

Let $G$ be the grid defined by the set of all axis-parallel lines
passing through the elements of $S$.
\begin{figure}[h]
  \centering
  \includegraphics[width=11.0cm]{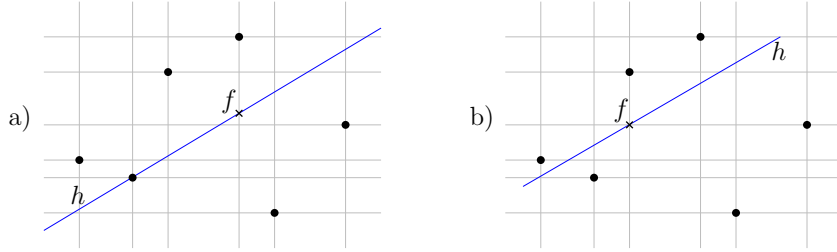}
  \caption{\small{a) Lemma~\ref{lemma:discretization-FFL}~(a). b)
Lemma~\ref{lemma:discretization-FFL}~(b)}}
  \label{fig:grid-and-disc}
\end{figure}
\begin{lemma}\label{lemma:discretization-FFL}
There always exists an optimal solution $(f,h)$ to the FFL-problem
satisfying one of the next statements: $(a)$ $h$ contains a point of
$S$ and $f$ is on a line of grid $G$, and $(b)$ $f$ is a vertex of
grid $G$. Refer to Fig.~\ref{fig:grid-and-disc}.
\end{lemma}
\begin{proof}
Let $(f,h)$ be an optimal solution of the FFL-problem satisfying
neither condition (a) nor condition (b). Using local linear
perturbations, we will transform $(f,h)$ into other optimal solution
that satisfies at least one of these conditions.

Assume first the case where $\varphi_v<\alpha\leq\frac{\pi}{4}$. Let
$\delta_1\geq 0$ (resp. $\delta_2\geq 0$) be the smallest value such
that if we translate both $f$ and $h$ with vector $(-\delta_1,0)$
(resp. $(\delta_2,0)$) then $f$ belongs to a vertical line of $G$ or
$h$ contains a point of $S$. Given $\eps\in[-\delta_1,\delta_2]$,
let $f_{\eps}$ and $h_{\eps}$ be $f$ and $h$ translated with vector
$(\eps,0)$, respectively. Using Lemma~\ref{lemma:way-of-move}, we
partition $S$ into four sets $Z_1$, $Z_2$, $Z_3$, and $Z_4$ as
follows. Set $Z_1$ (resp. $Z_2$) contains the demand points doing a
rightwards (resp. leftwards) movement to reach $h$. Set $Z_3$ (resp.
$Z_4$) contains the demand points doing only a downwards (resp.
upwards) movement to reach $h$. Observe that:
$$d_{h_{\eps}}(p,f_{\eps})-d_h(p,f)=\left\{
\begin{array}{cr}
  \eps & \text{ if }p\in Z_1\\
  -\eps & \text{ if }p\in Z_2\\
  \eps\cdot c_{\alpha} & \text{ if }p\in Z_3\\
  -\eps\cdot c_{\alpha} & \text{ if }p\in Z_4
\end{array}
\right.$$ where $c_{\alpha}=\tan\alpha-\frac{1}{\cos\alpha\cdot v}$.
Let $W_i$ denote $\sum_{p\in Z_i}w_p$, $(i=1,2,3,4)$. Thus, for any
$\eps\in[-\delta_1,\delta_2]$, the variation of objective function
when we translate both $f$ and $h$ with vector $(\eps,0)$ is the
following:
\begin{eqnarray*}
\Phi(f_{\eps},h_{\eps})-\Phi(f,h)
  &=& \sum_{p\in S}w_p\cdot\left(d_{h_{\eps}}(p,f_{\eps})- d_{h}(p,f)\right)\\
  &=& \sum_{p\in Z_1}w_p\eps+\sum_{p\in Z_2}w_p(-\eps)+\\
  & & \sum_{p\in Z_3}w_p(\eps c_{\alpha})+
      \sum_{p\in Z_4}w_p\left(-\eps c_{\alpha}\right)\\
  &=&\eps\left(W_1-W_2+c_{\alpha}(W_3-W_4)\right)
\end{eqnarray*}
Since $(f,h)$ is optimal we must have
$\Phi(f_{\eps},h_{\eps})-\Phi(f,h)\geq 0$ for all
$\eps\in[-\delta_1,\delta_2]$. It implies
$W_1-W_2+c_{\alpha}(W_3-W_4)=0$ and
$\Phi(f_{\eps},h_{\eps})=\Phi(f,h)$ for all
$\eps\in[-\delta_1,\delta_2]$. Therefore, by translating both $f$
and $h$ with vector either $(-\delta_1,0)$ or $(\delta_2,0)$, we
ensure that $f$ is on a vertical line of $G$ or $h$ passes through a
point of $S$, or both conditions. If it holds only that $f$ is on a
vertical line of $G$, then we repeat the same operation in the
vertical direction in order to ensure that $f$ is on a horizontal
line of $G$ or $h$ passes through a point of $S$, and condition (a)
or condition (b) holds. Otherwise, if it holds only that $h$ passes
through a point of $S$, then it is straightforward to prove, by
using similar arguments, that $f$ can be translated along $h$ in
order to ensure that $f$ belongs to a line of $G$, that is either vertical or
horizontal, and condition (a) holds. In fact, for some demand point
$p$ of $S$, $f$ will coincide with the point in which $p$ enters
$h$, that is, $p'$ or $p''$.

In the case where $0\leq\alpha\leq\varphi_v$ every demand point $p$
moves vertically to $h$ (Lemma~\ref{lemma:way-of-move}), and we can
proceed as follows by using arguments similar to the above ones. We
first translate vertically both $f$ and $h$ with the same vector in
order to ensure that $h$ contains a point of $S$. After that, $f$ is
translated along $h$ if necessary in order to $f$ belongs to a
vertical line of $G$ and condition (a) holds. The lemma thus
follows.
\end{proof}
\begin{corollary}\label{cor:discretization-FFL}
If $\alpha\leq\varphi_v$ then there is an optimal solution
satisfying Lemma~\ref{lemma:discretization-FFL}~$(a)$.
\end{corollary}

\section{The algorithm to solve the
FFL-problem}\label{section:algorithm-free}

\begin{theorem}\label{theorem:free/var/sum}
The FFL-problem can be solved in $O(n^3)$ time.
\end{theorem}
\begin{proof}
We find an optimal solution by solving two cases separately. The
first case is when solution satisfies
Lemma~\ref{lemma:discretization-FFL}~(a), and the second case is
when solution satisfies Lemma~\ref{lemma:discretization-FFL}~(b).

In order to solve the first case we find for each demand point $p$
and each line $\ell$ of $G$ an optimal angle $\alpha$ such that
$\Phi(f_{\alpha},h_{\alpha})$ is minimized, where $h_{\alpha}$ is
the line passing through $p$ whose angle with respect to the
positive direction of the $x$-axis is equal to $\alpha$, and
$f_{\alpha}$ is the intersection point between $\ell$ and
$h_{\alpha}$. Assume w.l.o.g.\ that $\ell$ is vertical and $p$ is
located to the left of $\ell$. It is easy to observe from
equations~(\ref{eq7}) and~(\ref{eq6}) that for any $\alpha\in[0,\frac{\pi}{4}]$ the
expression of $\Phi(f_{\alpha},h_{\alpha})$ has the form
$b_1+b_2\tan\alpha+b_3\cot\alpha+\frac{b_4}{\cos\alpha}+\frac{b_5}{\sin\alpha}$,
where $b_1,b_2,b_3,b_4,b_5$ are constants. Furthermore, if we progressively
increase the value of $\alpha$ from $0$ to $\frac{\pi}{4}$, that
expression changes whenever sets $S_1$, $S_2$, and $S_3$ change,
that is, when $h_{\alpha}$ crosses a demand point, $f_{\alpha}$
crosses a horizontal line of $G$, or $\alpha=\varphi_v$. Then
consider the sequence
$0=\alpha_0<\alpha_1<\dots<\alpha_m=\frac{\pi}{4}$ of $m+1=O(n)$
angles, where each angle $\alpha_i$ $(1\leq i<m)$ is such that
either $h_{\alpha_i}$ contains a demand point, $f_{\alpha_i}$
belongs to a horizontal line of grid $G$, or $\alpha_i=\varphi_v$.
Notice then that the expression of $\Phi(f_{\alpha},h_{\alpha})$ is
the same for all $\alpha\in[\alpha_i,\alpha_{i+1}]$ $(0\leq i<m)$.
If we preprocess the demand points $S$ by constructing the dual
arrangement of $S$~\cite{o-rourke98}, such a sequence can be
obtained in $O(n)$ time by using both the Zone
Theorem~\cite{o-rourke98} and the order of the demand points with
respect to the $y$-coordinate. Observe that if for a value of $i$ we
know the expression of $\Phi(f_{\alpha},h_{\alpha})$ in the interval
$[\alpha_i,\alpha_{i+1}]$, then $\Phi(f_{\alpha},h_{\alpha})$ can be
minimized in constant time in that interval. Furthermore, if
$h_{\alpha_{i+1}}$ contains a demand point or $f_{\alpha_{i+1}}$
belongs to a horizontal line of $G$, then the expression of
$\Phi(f_{\alpha},h_{\alpha})$ in the interval
$[\alpha_{i+1},\alpha_{i+2}]$ can be obtained in constant time from
the expression of $\Phi(f_{\alpha},h_{\alpha})$ in the interval
$[\alpha_i,\alpha_{i+1}]$. It is easy to see now that
$\Phi(f_{\alpha},h_{\alpha})$, $0\leq\alpha\leq\frac{\pi}{4}$, can
be minimized in $O(n)$ time by minimizing
$\Phi(f_{\alpha},h_{\alpha})$ in $[\alpha_i,\alpha_{i+1}]$ for
$i=0,1,\dots,m-1$. Since there are $n$ demand points and $G$ has
$2n$ lines, then an overall $O(n^3)$-time algorithm is obtained.

We can proceed similarly in order to solve the second case. We find
an optimal solution $(u,h)$ for each vertex $u$ of the grid $G$ as
follows. Let $u$ be a vertex of $G$. Given an angle $\alpha$, let
$h_{\alpha}$ be the line passing through $u$, whose angle with
respect to the positive direction of the $x$-axis is equal to
$\alpha$. Then, by Corollary~\ref{cor:discretization-FFL}, we look
for an angle $\alpha\in(\varphi_v,\frac{\pi}{4}]$ such that the
objective function $\Phi(u,h_{\alpha})$ is minimized. It follows
from equation~(\ref{eq6}) that the expression of
$\Phi(u,h_{\alpha})$ has the form
$c_1+c_2\tan\alpha+c_3\cot\alpha+\frac{c_4}{\cos\alpha}+\frac{c_5}{\sin\alpha}$ for any
$\alpha\in(\varphi_v,\frac{\pi}{4}]$, where $c_1,\ldots,c_5$ are
constants. If we progressively increase $\alpha$ from $\varphi_v$ to
$\frac{\pi}{4}$ the expression of $\Phi(u,h_{\alpha})$ keeps
unchanged as long as $h_{\alpha}$ does not cross a demand point. The
sorted sequence of values of $\alpha$ in which it happens can be
obtained in linear time by using duality~\cite{o-rourke98}. That
sequence of values induces a partition of the interval
$(\varphi_v,\frac{\pi}{4}]$ into intervals where in each of them the
expression of $\Phi(u,h_{\alpha})$ is constant. We can now continue
as was done above to solve the first case. Since $G$ has $O(n^2)$
vertices then an overall $O(n^3)$-time algorithm is thus obtained.
The result thus follows.
\end{proof}

\section{A refinement of the
algorithm}\label{section:refinement}

 Theorem~\ref{theorem:free/var/sum} shows an algorithm
that solves the FFL-problem by dividing the search of optimal
solutions into two steps. It first looks in $O(n^3)$ time for an
optimal solution that satisfies
Lemma~\ref{lemma:discretization-FFL}~(a), and after that looks
within the same time complexity for an optimal solution satisfying
Lemma~\ref{lemma:discretization-FFL}~(b). In the following we show
that for ``reasonable'' values of speed $v$ we can simplify the
algorithm of Theorem~\ref{theorem:free/var/sum} by finding only
optimal solutions that hold
Lemma~\ref{lemma:discretization-FFL}~(a). We will use the next
technical lemma.
\begin{lemma}\label{lemma:function}
Let $a$, $b$, $c$, and $v>\frac{3\sqrt{2}}{4}$ be non-negative
constants and $F:(0,\frac{\pi}{2})\rightarrow\mathbb{R}$ be a
function so that
$$F(x)=a\left(\frac{1-v\sin x}{\cos x}\right)+b\left(\frac{1-v\cos x}{\sin x}\right)+c$$
for all $x\in(0,\frac{\pi}{2})$. Then next statements are true:
\begin{itemize}
  \item[$($a$)$] If $a=0$ and $b=0$ then $F$ is constant.
  \item[$($b$)$] If $a>0$ and $b=0$ then $F$ is monotone decreasing.
  \item[$($c$)$] If $a=0$ and $b>0$ then $F$ is monotone increasing.
  \item[$($d$)$] If $a>0$ and $b>0$ then $F$ has no minima.
\end{itemize}
\end{lemma}
\begin{proof}
Statement (a) is immediate. Let $F'$ be the first derivative of $F$
and observe that:
$$F'(x)=a\left(\frac{\sin x-v}{\cos^2x}\right)+b\left(\frac{v-\cos x}{\sin^2x}\right)$$
If $a>0$ and $b=0$ then $\lim_{x\rightarrow 0^{+}}F(x)=a$,
$\lim_{x\rightarrow \frac{\pi}{2}^{-}}F(x)=-\infty$, and equation
$F'(x)=0$ has no solution in $(0,\frac{\pi}{2})$ because $v>1$.
Therefore, $F(x)$ is a monotone decreasing function and statement
(b) thus holds.

If $a=0$ and $b>0$ then $\lim_{x\rightarrow 0^{+}}F(x)=-\infty$,
$\lim_{x\rightarrow \frac{\pi}{2}^{-}}F(x)=b$, and equation
$F'(x)=0$ has no solution in $(0,\frac{\pi}{2})$ because $v>1$.
Therefore, $F(x)$ is a monotone increasing function and statement (c) thus
holds.

Consider $a>0$ and $b>0$. Since $\lim_{x\rightarrow
0^{+}}F(x)=\lim_{x\rightarrow \frac{\pi}{2}^{-}}F(x)=-\infty$ it
suffices to prove that equation $F'(x)=0$ has only one solution
which must be a global maximum of $F$.
Equation $F'(x)=0$ is equivalent to equation
$$G(x):=\frac{\sin^2x(v-\sin x)}{\cos^2x(v-\cos x)}=\frac{b}{a}$$
We will prove that equation $G'(x)=0$ has no solution in
$(0,\frac{\pi}{2})$, which implies that equation $G(x)=\frac{b}{a}$
has a unique solution in $(0,\frac{\pi}{2})$ because
$\lim_{x\rightarrow
0^+}G(x)=0<\frac{b}{a}<+\infty=\lim_{x\rightarrow
\frac{\pi}{2}^-}G(x)$. This will complete the proof.

It is straightforward to see that equation $G'(x)=0$ is equivalent
to equation $H(x):=0$, where $H(x)$ is equal to:
$$2(v-\sin x)(v-\cos x)+\sin x\cos x(1-v(\sin x+\cos x))$$
Consider the function $I:(0,\frac{\pi}{2})\rightarrow\mathbb{R}$ so
that
$$I(x)=2(v-\sin x)(v-\cos x)+\sin x\cos x(1-\sqrt{2}v)$$
for all $x\in(0,\frac{\pi}{2})$. Since $\sin x+\cos x\leq\sqrt{2}$
and $\sin x\cos x>0$ for all $x\in(0,\frac{\pi}{2})$, we have
$H(x)\geq I(x)$ for all $x\in(0,\frac{\pi}{2})$. We now show that
the minimum of $I(x)$ in $(0,\frac{\pi}{2})$ is greater than zero,
implying equation $H(x)=0$ has no solution.
\begin{eqnarray*}
I'(x)&=&2(-\cos x(v-\cos x)+\sin x(v-\sin x))+\\
& &(1-\sqrt{2}v)(\cos^2x-\sin^2x)\\
&=&2(\cos x-\sin x)(\cos x+\sin x-v)+\\
& &(1-\sqrt{2}v)(\cos x-\sin x)(\cos x+\sin x)\\
&=&(\cos x-\sin x)\left((3-\sqrt{2}v)(\sin x+\cos x)-2v\right)
\end{eqnarray*}
If $3-\sqrt{2}v<0$ then there is no $x\in(0,\frac{\pi}{2})$ such
that $(3-\sqrt{2}v)(\sin x+\cos x)-2v=0$ because $\sin x+\cos x$ is
positive for all $x\in(0,\frac{\pi}{2})$. Suppose
$3-\sqrt{2}v\geq0$. Then we have:
\begin{eqnarray*}
 (3-\sqrt{2}v)(\sin x+\cos x)-2v & \leq & (3-\sqrt{2}v)\sqrt{2}-2v\\
 & = & 3\sqrt{2}-4v\\
 & < & 0
\end{eqnarray*}
and thus there is no $x\in(0,\frac{\pi}{2})$ such that
$(3-\sqrt{2}v)(\sin x+\cos x)-2v=0$. Therefore, $I'(x)=0$ if and
only if $\cos x-\sin x=0$, that is, $x=\frac{\pi}{4}$. Let us prove
that $I(\frac{\pi}{4})>0$.
$$I\left(\frac{\pi}{4}\right)=2\left(v-\frac{\sqrt{2}}{2}\right)^2+\frac{1}{2}
\left(1-\sqrt{2}v\right)=2v^2-\frac{5\sqrt{2}}{2}v+\frac{3}{2}$$ The
roots of polynomial $P(x):=2x^2-\frac{5\sqrt{2}}{2}x+\frac{3}{2}$
are respectively equal to
$\frac{1}{4}\left(\frac{5\sqrt{2}}{2}-\frac{\sqrt{2}}{2}\right)$ and
$\frac{1}{4}\left(\frac{5\sqrt{2}}{2}+\frac{\sqrt{2}}{2}\right)=\frac{3\sqrt{2}}{4}$. 
Then we conclude that
$P(v)=I\left(\frac{\pi}{4}\right)>0$ because the main coefficient of
$P(x)$ is positive, $\frac{3\sqrt{2}}{4}$ is the greatest root of
$P(x)$, and $v>\frac{3\sqrt{2}}{4}$. Since $\lim_{x\rightarrow
0^{+}}I(x)=\lim_{x\rightarrow \frac{\pi}{2}^{-}}I(x)=2v(v-1)>0$ and
$I(x)>0$ at the unique extreme point $x=\frac{\pi}{4}$, then
$I(x)>0$ for all $x\in(0,\frac{\pi}{2})$. Therefore, $H(x)>0$ for
all $x\in(0,\frac{\pi}{2})$, equation $G'(x)$ has no solution in
$(0,\frac{\pi}{2})$, and then $F(x)$ has only one extreme point in
$(0,\frac{\pi}{2})$ which is a global maximum. The lemma
follows.
\end{proof}

\begin{lemma}\label{lemma:discretization-FFL-2}
If speed $v$ is greater than $\frac{3\sqrt{2}}{4}\approx
1.060660172$, then there always exists an optimal solution $(f,h)$
to the FFL-problem satisfying
Lemma~\ref{lemma:discretization-FFL}~$(a)$, that is, $h$ contains a
point of $S$ and $f$ is on a line of grid $G$.
\end{lemma}
\begin{proof}
It suffices to show that for every solution $(f,h)$, there
exists another solution $(f',h')$ satisfying
Lemma~\ref{lemma:discretization-FFL}~(a) and
$\Phi(f',h')\leq\Phi(f,h)$. The proof is as follows.

Let $(f,h)$ be a solution of the FFL-problem so that $h$
contains no demand point. Assume here angle $\alpha$ satisfies
$0\leq\alpha<\frac{\pi}{2}$. If $\alpha$ is such that
$0\leq\alpha\leq\varphi_v$, then the result follows from
Corollary~\ref{cor:discretization-FFL}. Therefore, assume
$\varphi_v<\alpha\leq\frac{\pi}{2}$.
We proceed to prove that there exists another solution
$(f',h')$ such that $h'$ contains a demand point and
$\Phi(f',h')\leq\Phi(f,h)$. Consider the sets $S_2=S_2(f,h)$ and
$S_3=S_3(f,h)$ as defined before.

Given an angle $\beta\in[\varphi_v,\frac{\pi}{2})$, let $h_{\beta}$
be the line passing through $f$ so that the angle between
$h_{\beta}$ and the positive direction of the $x$-axis is equal to
$\beta$. Observe then by equation~(\ref{eq6}) that $\Phi(f,h)$ is equal to:
$$\Phi(f,h_{\alpha})=a\left(\frac{1-v\sin\alpha}{\cos\alpha}\right)+b\left(\frac{
1-v\cos\alpha}{\sin\alpha}\right)+c$$
where $a$, $b$, and $c$ are non-negative constants that depend on
the coordinates of the demand points and the speed $v$. Then we can
argue what follows by both noting that $a>0$ (resp. $b>0$) if and only if
$S_2$ (resp. $S_3$) is not empty and using
Lemma~\ref{lemma:function}.

We can rotate $h$ with center $f$ by either increasing or decreasing
$\alpha$ to the value $\alpha'\in[\varphi_v,\frac{\pi}{2})$ in such
a way solution $(f,h_{\alpha'})$ is obtained, where either
$\alpha'=\varphi_v$ or $h_{\alpha'}$ contains a demand point of
$S_2\cup S_3$, and $\Phi(f,h_{\alpha'})\leq\Phi(f,h)$. If
$\alpha'=\varphi_v$ then the result follows from
Corollary~\ref{cor:discretization-FFL}. Otherwise, if $h_{\alpha'}$
contains a demand point of $S_2\cup S_3$, then
$(f',h')=(f,h_{\alpha'})$ is the desired solution, and to
finalize the proof, we translate $f'$ along $h'$, if necessary, as
was done in the proof of Lemma~\ref{lemma:discretization-FFL}, in
order to ensure $f'$ belongs to a line of grid $G$. The result thus
follows.
\end{proof}

\section{Examples}\label{section:experimental}

In Fig.~\ref{fig:experiment1} and Fig.~\ref{fig:experiment2} we show
the same example, consisting of nine demand points
$p_1,p_2,\ldots,p_9$. Each demand point is represented by a solid
dot, and labeled with a triple, the first two components are the
coordinates and the third component is its weight. Facility point
$f$ is represented by a cross. In both examples optimal solutions
satisfy Lemma~\ref{lemma:discretization-FFL}~(a), that is, highway
contains a demand point and facility point belongs to a line of grid
$G$. As expected, when we increase the highway's speed from the
example in Fig.~\ref{fig:experiment1} to the one in
Fig.~\ref{fig:experiment2}, the shortest paths to the facility
point change according to the claims of
Lemma~\ref{lemma:way-of-move}.
\begin{figure}[h]
  \centering
  \includegraphics[width=10cm]{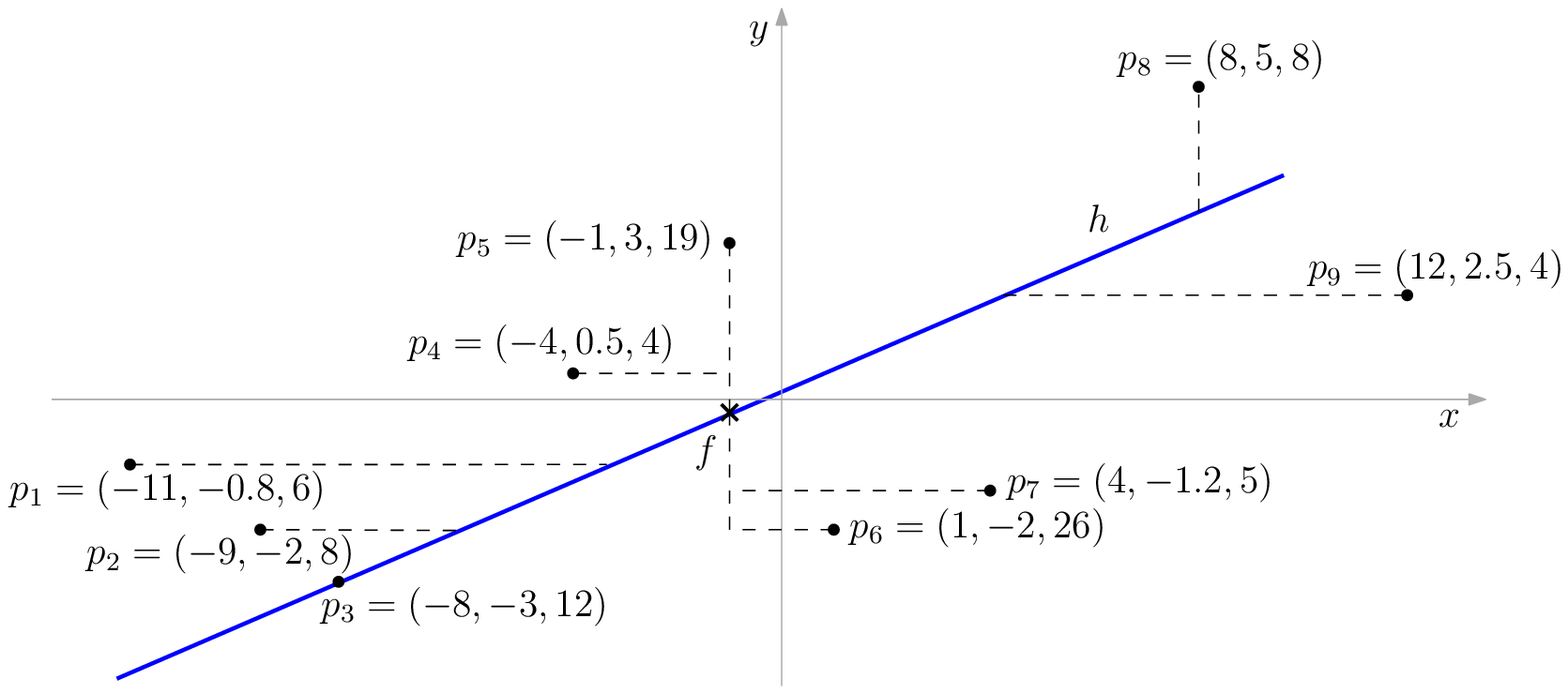}
  \caption{\small{The highway contains a demand point and facility point belongs
  to a line of grid $G$. Some points perform an horizontal movement to reach the
highway.}}
  \label{fig:experiment1}
\end{figure}

\begin{figure}[h]
  \centering
  \includegraphics[width=10cm]{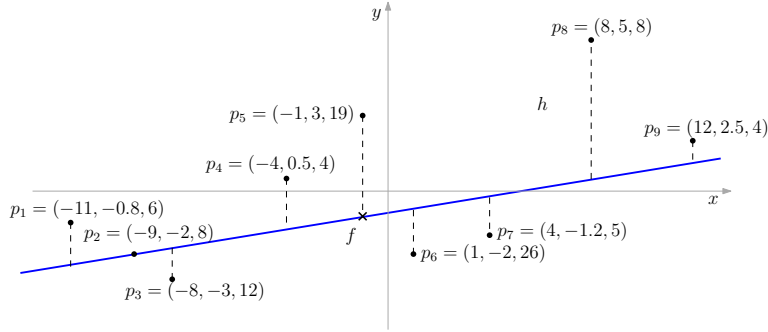}
  \caption{\small{The highway contains a demand point and facility point belongs
  to a line of grid $G$. All points perform only a vertical movement to reach the
highway.}}
  \label{fig:experiment2}
\end{figure}

In Fig.~\ref{fig:experiment1} speed $v$ is equal to $1.2$ and
highway $h$ of the optimal solution contains point $p_3$, facility
point $f=(-1,-0.14335)$ is on the vertical line passing through
$p_5$, and there are some demand points moving horizontally to reach
facility point $f$. The value of the objective function $\Phi(f,h)$
is equal to $525.83$.

In Fig.~\ref{fig:experiment2} speed $v$ is equal to $1.5$, highway
$h$ of the optimal solution contains point $p_2$, facility
$f=(-1,-0.7758)$ is on the vertical line containing $p_5$, and all
demand points move vertically only to reach the highway. Since speed
is greater than speed in Fig.~\ref{fig:experiment1}, the value of
the objective function $\Phi(f,h)$ reduces to $471.55$.

In Fig.~\ref{fig:experiment3} we present a different example
consisting of nine demand points $p_1,\ldots,p_9$, with the aim of
showing the existence of configurations for which optimal solutions
satisfy only Lemma~\ref{lemma:discretization-FFL}~(b). In this
example speed $v$ is equal to $1.04<\frac{3\sqrt{2}}{4}$. Highway
$h$ of the optimal solution contains no demand point and facility
$f=(0,0)$ is on a vertex of grid $G$, in fact, it is located on both
the horizontal line through $p_7$ and the vertical line through
$p_6$. The value of $\Phi(f,h)$ is equal to $336.2$.

\begin{figure}[h]
  \centering
  \includegraphics[width=9cm]{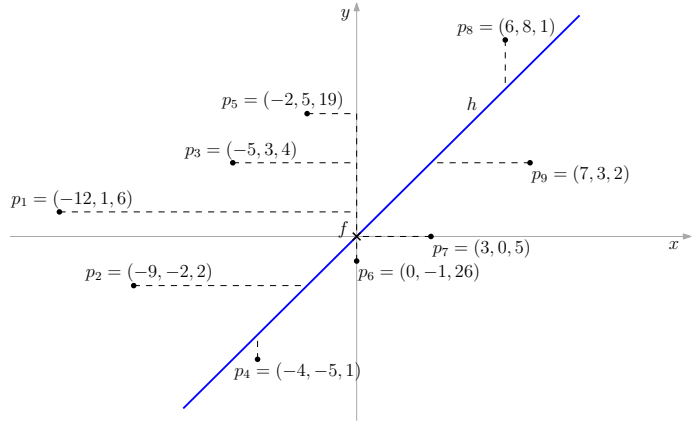}
  \caption{\small{The highway contains no demand point and facility point
  is a vertex of grid $G$.}}
  \label{fig:experiment3}
\end{figure}

\section{Conclusions and further research}\label{section:conclu}

We have solved in $O(n^3)$ time the problem of locating at the same
time a facility point and a freeway of variable
length,
among a set of demand points, in order to minimize the total
weighted travel time from the demand points to the facility.
Some examples are presented to show that there exist optimal
solutions corresponding to each type of solutions that our
algorithm considers.

A natural restriction to be considered in further research
of this problem is to upper bound the length of the highway,
that is, to consider that highway has fixed length.
In this case, there also exist optimal solutions in which the facility
point belongs to the highway. This was in fact showed in
Proposition~\ref{prop:fac-in-highway} because in the proof
we did not change the lenght of the highway.
It is not hard to see that when the highway's length is fixed,
the shortest paths from the demand points to the facility point can be
discretized as follows. If $0\leq\alpha\leq\varphi_v$, then we distinguish
three regions $R_1$, $R_2$, and $R_3$ as depicted in
Fig.~\ref{fig:disc-fixed-len}~a). Points belonging to $R_1\cup R_3$
move to the nearest endpoint of $h$,
and points of $R_2$ move vertically to $h$. Otherwise,
if $\varphi_v<\alpha\leq\frac{\pi}{4}$, then eight regions
$R_1,\ldots,R_8$ can be identified as 
shown in Fig.~\ref{fig:disc-fixed-len}~b). Points of
$R_1\cup R_2$ move to the nearest endpoint of $h$, 
points of $R_3\cup R_4$ move directly to $f$, points of $R_5\cup R_6$
move horizontally to $h$, and points of $R_7\cup R_8$ move vertically
to $h$. 
\begin{figure}[h]
  \centering
  \includegraphics[width=15cm]{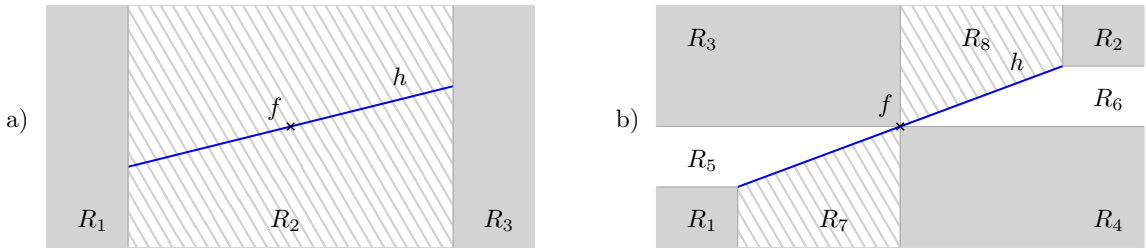}
  \caption{\small{Discretization when the highway's length is fixed.}}
  \label{fig:disc-fixed-len}
\end{figure}

We believe that with the above discretization, the search space of
optimal solutions can be simplified by using a similar (and more detailed)
statement as Lemma~\ref{lemma:discretization-FFL}. This will permit
to obtain an algorithm similar to the one presented 
in Theorem~\ref{theorem:free/var/sum}.

Other variant to be considered in further research is the problem of locating at the
same time the facility point and a turnpike of variable length. This will
extend~\cite{diaz-banez11,espejo11}.

\section*{Acknowledgement}

Authors D\'iaz-B\'a\~nez and Ventura were partially supported by 
project FEDER MEC MTM2009-08652 and ESF EUROCORES programme 
EuroGIGA - ComPoSe IP04 - MICINN 
Project EUI-EURC-2011-4306. 
Korman had the support of the Secretary for 
Universities and Research of the Ministry of Economy and Knowledge of the 
Government of Catalonia and the European Union.
P\'erez-Lantero was partially supported by project FEDER MEC MTM2009-08652 and
grant FONDECYT 11110069.

\nocite{*}

\small

\bibliographystyle{plain}
\bibliography{freeway}

\end{document}